\documentclass{llncs}
\usepackage{xspace}
\usepackage{comment}
\usepackage{subfigure}
\usepackage{mathrsfs}

\usepackage[pdftex]{graphicx}
\usepackage[pdftex]{color}
\usepackage{array}
\usepackage{comment}
\usepackage{stmaryrd}

\usepackage{tikz}
\tikzstyle{uredge}=[draw=red]
\tikzstyle{redge}=[draw=green!50!black]
\tikzstyle{rnode}=[]
\def\miko{Mikol\'a\v{s} Janota}
\def\theTitle{%
      An {Achilles'} Heel of Term-Resolution%
}
\title{\theTitle}
\author{{\miko}\inst{1}}
\institute{IST/INESC-ID, Lisbon, Portugal}

\usepackage{amsmath,amssymb}

\DeclareMathOperator*{\Neg}{\textrm{Neg}}
\DeclareMathOperator*{\nand}{\textsf{NAND}}

\DeclareMathOperator*{\Pref}{{\mathcal P}}

\DeclareMathOperator*{\var}{{\sf var}}
\DeclareMathOperator*{\lev}{{\sf lv}}
\DeclareMathOperator*{\vars}{{\sf var}}
\DeclareMathOperator*{\union}{\cup}
\DeclareMathOperator*{\intersection}{\cap}

\newcommand{\comprehension}[2]{\ensuremath{\left\{ {#1} \;|\; {#2}\right\}}}
\usepackage{bussproofs}

\usepackage[linesnumbered,ruled,vlined]{algorithm2e}
\DontPrintSemicolon%
\SetKwInOut{Input}{in}
\SetKwInOut{Output}{out}
\usepackage{color}
\definecolor{citeblue}{rgb}{0.1,0,.4}
\usepackage[pdftex%
,colorlinks=true%
,bookmarks=true%
,linkcolor=citeblue%
,citecolor=citeblue%
,urlcolor=blue%
,plainpages=false]{hyperref}
\hypersetup{%
pdfauthor={ {\miko} },
pdftitle={\theTitle}
}

\fontfamily{ptm}

\begin{document}
\maketitle
\begin{abstract}
  Term-resolution provides an elegant mechanism to prove that a
  quantified Boolean formula (QBF) is true. It is a dual
  to Q-resolution (also referred to as clause-resolution) and is
  practically highly important as it enables certifying answers of
  DPLL-based QBF solvers. While term-resolution and Q-resolution are
  very similar, they're not completely symmetric. In particular,
  Q-resolution operates on clauses and term-resolution operates on
  models of the matrix.  This paper investigates what impact this
  asymmetry has. We'll see that there is a large class of formulas
  (formulas with ``big models'') whose term-resolution proofs are
  exponential. As a possible remedy, the paper suggests to prove true
  QBFs by refuting their negation ({\em negate-refute}), rather than
  proving them by term-resolution. The paper shows that from the
  theoretical perspective this is indeed a favorable approach.  In
  particular, negation-refutation can p-simulates term-resolution and
  there is an exponential separation between the two calculi.  These
  observations further our understanding of proof systems for QBFs and
  provide a strong theoretical underpinning for the effort towards
  non-CNF QBF solvers.
\end{abstract}

\section{Introduction}\label{sec:introduction}
Arguably, the interest of computer scientists in proof complexity
begins with the seminal work of Cook and Reckhow who showed a relation
between proof complexity and the question NP vs.\
co-NP~\cite{DBLP:journals/jsyml/CookR79}. This interest was further
fueled by the practical success of programs for automated reasoning,
such as {\em SMT solvers} or {\em SAT solvers}. Machine-verifiable
proofs serve as {\em certificates} for such solvers. It is important
that a solver can produce a certificate of its answer as the solver
itself can contain
bugs~\cite{VG02-cfv,DBLP:conf/sat/BrummayerLB10}. Moreover, proofs
have turned out to be important artifacts for further computations,
like invariant inference for example~\cite{DBLP:conf/cav/McMillan03}.
This paper follows this line of research, i.e.\ proof complexity and
solver complication certification, with the focus on {\em quantified
Boolean formula} (QBF). In particular, it focuses on QBFs whose propositional
part is in {\em conjunctive normal form} (QCNF). QCNF is complete and widely
popular input for QBF solvers due to its susceptibility to simple
representation inside the solver.

A number of QCNF solvers take inspiration in the approach that turned
out to be so successful for SAT; and that is {\em conflict driven
  clause learning}
(CDCL)~\cite{DBLP:conf/ictai/SilvaS96,DBLP:series/faia/SilvaLM09}.
Since {\em propositional resolution} is the underlying proof principle
used in SAT, an analogous proof system was developed for QCNF. In
particular, {\em Q-resolution}~\cite{DBLP:journals/iandc/BuningKF95}
for false formulas, and {\em
  term-resolution}~\cite{GiunchigliaEtAlJAIR06} for true formulas.
It has been shown that CDCL-based QBF
solvers~\cite{DBLP:conf/iccad/ZhangM02} can be certified by these two
proof systems~\cite{GiunchigliaEtAlJAIR06}. Recently, several proof
complexity analyses of Q-resolution were published.  A separation
result for Q-resolution and a sequent calculus by
Kraj{\'\i}{\v{c}}ek~and~Pudl{\'a}k~\cite{KrajicekPudlak90} is shown by
Egly~\cite{DBLP:conf/sat/Egly12}; Van Gelder shows that enabling
resolution on universal-variables in Q-resolution proofs gives an
exponential advantage to Q-resolution~\cite{DBLP:conf/cp/Gelder12};
Janota and Silva show some p-simulation results for fragments of
Q-resolution and solving QBF by expanding universal
variables~\cite{JanotaSilva-SAT13}.

This paper brings the focus to term-resolution.  While term-resolution
is an elegant system because it provides a dual to Q-resolution, the
two types of resolution are not perfectly symmetric. This is because
Q-resolution can operate on the given clauses but term-resolution
operates on the satisfying assignments of those clauses. This paper
shows that this difference exposes an Achilles' heel of
term-resolution.

The first result of this paper is that it shows that term-resolution
proofs are large for QCNFs whose propositional part have models with a
large number of universal literals. More precisely, if each model has
at least $k$ universal literals, any term-resolution proof has at
least $2^k$ nodes. Subsequently, the paper investigates an alternative
route to term-resolution and that is {\em refuting the negation} of
the formula. The paper shows that any term-resolution proof can be
translated to a negation-refutation in polynomial time. On a
particular formula $\Psi$ we show an exponential separation between
negation-refutation and term-resolution, i.e.\ all term-resolution
proofs of $\Psi$ are exponential but there is a Q-resolution proof of
$\lnot\Psi$ is polynomial.

These results have direct practical implications for QBF solving
because term-resolution enables certifying DPLL-based QBF
solvers. Consequently, a formula whose term-resolution proofs are
exponential, will require exponential time to {\em solve}.  These
theoretical results further substantiate an observation already made
in the QBF community and that is that QBF with propositional part in
CNF are particularly harmful for solving
solving~\cite{DBLP:conf/aaai/AnsoteguiGS05,DBLP:conf/aaai/Zhang06}.
\section {Preliminaries}\label{section:introduction}
A {\em literal} is a Boolean variable or its negation.  For a literal
$ l $, we write $\bar l $ to denote the literal {\em complementary} to
$ l $, i.e.\ $\bar x=\lnot x$, $\overline{\lnot x}=x$.  A {\em clause}
is a disjunction of finitely many literals.  A formula in {\em
  conjunctive normal form} (CNF) is a conjunction finitely many
clauses. As common, whenever convenient, a clause is treated as a set
of literals and a CNF formula as a set of sets of literals.

For a literal $l=x$ or $l=\bar x$, we write $\vars(l)$ for $x$.  
Analogously, for
a clause $C$,  $\vars(C)$  denotes
$\comprehension{\vars(l)}{l\in C}$ and for a CNF $\psi$, $\vars(C)$
denotes $\comprehension{l}{l\in\vars(C), C\in\psi}$

For a set of variables $ X $ an {\em assignment} is a function from $
X $ to the constants $ 0 $ and $ 1 $. We say that the assignment is
{\em complete} for $ X $ if the function is total.

Analogously to a clause, a {\em term} is a conjunction of finitely many
literals. Again, whenever convenient, a term is treated as a set of
literals. We say that a term $ T $ is a {\em model} of a CNF $\phi $
if and only if for each clause $C\in\phi$ there is a literal $ l $
both in $T$ and $C$, i.e.\ $C\intersection T\neq\emptyset$.

There is an obvious relation between terms and assignments.  A term
uniquely determines a set of assignments that satisfy the term. If
an assignment satisfies a model of $\phi $, then it is a satisfying
assignment.  Note that some definitions require a model to be a
complete assignment to the variables of $\phi$. The aforementioned
correspondence shows that there's no substantial difference between
the definitions.

{\em Quantified Boolean Formulas} (QBFs)~\cite{DBLP:series/faia/BuningB09} are an extension of propositional
logic with quantifiers with the standard semantics that 
$\forall x.\,\Psi$ is satisfied by the same truth assignments as  $\Psi[x/0]\land\Psi[x/1]$
and
$\exists x.\,\Psi$ as  $\Psi[x/0]\lor\Psi[x/1]$.
 Unless specified otherwise,   QBFs are in 
{\em  closed} {\em  prenex}  form with a CNF {\em matrix}, i.e.\
\hbox{${\cal Q}_1 X_1 \dots {\cal Q}_k X_k.\ \phi$},
where 
$X_i$ are pairwise disjoint sets of variables; 
${\cal Q}_i\in\{\exists,\forall\}$ and ${\cal Q}_i\neq {\cal Q}_{i+1}$.
The formula $\phi$  is in CNF and is  defined only on variables  $X_1\union\dots\union X_k$.
 The propositional part $\phi $ is called the {\em matrix} and the rest the {\em prefix}.
We write QCNF to talk about formulas in this form.
If a variable $x$ is in the set $X_i$, we say that $x$ is at {\em
  level} $i$ and write $\lev(x)=i$; we write $\lev(l)$ for
$\lev(\var(l))$.  A closed QBF is {\em false} (resp.\ {\em true}), iff
it is semantically equivalent to the constant $0$ (resp.\ $1$).

 If a variable is universally quantified, we say that the variable is
 universal.  For a literal $l$ and a universal variable $ x $ such
 that $\vars(l)=x$, we say that $ l $ is universal. The notions of
 existential variable and term are defined analogously.

\subsection{Q-resolution}\label{sec:Qresolution}
Q-resolution~\cite{DBLP:journals/iandc/BuningKF95} is an extension of
propositional resolution for showing that a QCNF is false.
For a clause $C$, a universal literal $l\in C$ is {\em blocked} by an
existential literal $k\in C$ iff $\lev(l)<\lev(k)$.  {\em
  $\forall$-reduction} is the operation of removing from a clause $C$
all universal literals that are {\em not} blocked by some literal.
For two $\forall$-reduced clauses $x\lor C_1$ and $\bar x\lor C_2$, where $x$ is an
existential variable, a {\em Q-resolvent}~\cite{DBLP:journals/iandc/BuningKF95} is
obtained in two steps.  (1)~Compute  $C_u = C_1\union C_2\smallsetminus\{x,\bar x\}$.
If $C_u$ contains complementary literals, the Q-resolvent is undefined.
(2)~$\forall$-reduce $C_u$.
For a QCNF $\Pref. \phi$, a A {\em Q-resolution proof} of a clause $C$
is a finite sequence of clauses $C_1,\dots,C_n$ where $C_n=C$ and any $C_i$
in the sequence is part of the given matrix $\phi$ or it is a
Q-resolvent for some pair of the preceding clauses.  A Q-resolution
proof is called a {\em refutation} iff $C$ is the empty clause, denoted
$\bot$.
\vspace{2pt}

\noindent
\begin{minipage}{.78\textwidth}
\quad In this paper Q-resolution proofs treated as connected directed
acyclic graphs so that the each clause in the proof corresponds to
some node labeled with that clause.  We assume that
the input clauses are already $\forall$-reduced.
 Q-resolution steps are depicted as on the right.
Note that the $\forall$-reduction step is depicted separately.
\end{minipage}
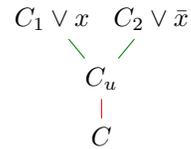
\tikzstyle{level 1}=[level distance=15pt, sibling distance=1.5cm]
\tikzstyle{level 2}=[level distance=15pt, sibling distance=1.13cm]
\begin{tikzpicture}[baseline=15pt,grow=up]
\node[rnode,below] (r) at (0,0) {$C$} ;
\node[rnode] (ru) at (0,15pt)  { $C_u$ } ;
\node[rnode,above] (p1) at (-.66,30pt) {$C_1\lor x$} ;
\node[rnode,above] (p2) at (0.66,30pt) {$C_2\lor\bar x$} ;
\draw[redge] (p1) -- (ru) -- (p2) ;
\draw[uredge] (r) -- (ru) ;
\end{tikzpicture}

\subsection{Term-Resolution}
\label{sec:tresd}

{\em Term-resolution} is analogous to Q-resolution with the difference
that it operates on terms and its purpose is to prove that a QCNF is
true~\cite{GiunchigliaEtAlJAIR06}. Since the calculus operates on
QBF's with CNF matrices, it needs a mechanism to generate terms to
operate on. This is done by a rule that enables using models of the
given matrix in the proof.

Term-resolution, resolves on universal literals and reduces
existential ones.  For a term $ T $ an existential literal $l$ is {\em
  blocked}, iff there is a universal $k\in T$ such that
$\lev(l)<lev(k)$.  $\exists$-reduction removes from a term $T$ all
existential literals that are {\em not} blocked by some universal
literal. For two $\exists $-reduced terms $x\land T_1$ and $\bar x\land
T_2$, a {\em term-resolvent} is defined as the $\exists$-reduction of
the term $T_1\land T_2$, if $T_1$ and $T_2$ do not contain
complementary literals; it is undefined otherwise.

For a QCNF $P.\ \phi$ a {\em term-resolution proof} of the term $T_m$
a is a finite sequence $T_1,\dots,T_m$ of terms such that each term
$T_i$ is a model of $\phi $ or it was obtained from the previous terms
by $\exists $-reduction or term-resolution. Such proof {\em proves}
$\Pref.\ \phi$ iff $T_m$ is the empty term, denoted as $\top $.  Those
terms of the proof that are models of $\phi $ are said to be generated
by a {\em model generation rule}.  (Terms are sometimes referred to as
``cubes'', especially in the context of DPLL QBF solvers that apply
cube learning.)
%

\subsection{Proof complexity}
\label{sec:proof:complexity}

A proof system $P$ is relation $P(\Phi,\pi)$ computable in polynomial
time such that a formula $\Phi$ is true iff there exists a proof $\pi$
for which $P(\Phi,\pi)$.
A proof system $P_1$ {\em p-simulates} a proof system $P_2$ iff any
proof in $P_2$ of a formula $\Phi$ can be translated into a proof in
$P_1$ of $\Phi$ in polynomial time
(cf.~\cite{DBLP:journals/jsyml/CookR79,DBLP:journals/eatcs/Urquhart98}).

As is common, we will count the sizes of Q-resolution and
term-resolution as the number of resolution steps and number of
$\forall$/$\exists$-reductions where each reduced literal is counted
separately.

\section {The Achilles' Heel}\label{section:heel}

This section describes a large class of formulas that have exponential
term-resolution proofs. Recall that a leaf of a term-resolution proof
must be generated by the model-generation rule. We exploit this by
forcing the proof to generate many leafs.

First we make a simple observation that for any assignment to
universal variables, there must be a leaf-term in a term-resolution
proof that ``agrees'' with that assignment.  We say that a term
$T$ {\em agrees} with an assignment $\tau$ iff there is no literal $l$
such that $\bar l\in T$ and $\tau(l)=1$.

\begin{lemma}\label{lemma:accordance}
  For any assignment $\tau$ to universal variables and a
  term-resolution proof $\pi$ there is a leaf-term $T$ of $\pi$ that
  agrees with $\tau $.
\end{lemma}
\begin{proof}
  We construct a path from the root to some leaf such that each node
  on that path agrees with $\tau$.  The root of $\pi$ agrees
  with $\tau$ because it does not contain any literals. If a term $T$
  agrees with $\tau$ and $T$ is obtained from $T'$ by
  existential-reduction, then $T'$ also agrees with $\tau$ since $\tau
  $ assigns only to universal variables. If $ T $ agrees with $\tau $
  and is obtained from $T_1$ and $T_2$ by term-resolution on some
  variable $y$, it has to be that $y$ is in one of the $ T_1 $, $ T_2 $
  and $\bar y$ in the other. Hence, at least one of the terms agrees
  with $\tau$.  \qed
\end{proof}

\begin{theorem}\label{theorem:k}
  If all models of $\phi $ contain at least $k$ universal literals,
  then any term-resolution proof of $\Phi$ has at least $2^k$
  leafs. (Recall that a formula has a term-resolution proof if and
  only if it is true.)
\end{theorem}
\begin{proof}
  Let $V_u$ be the set of universal variables of $\Phi$. Since each
  leaf-term of any term-resolution proof has at least $ k $ universal
  literals, it can agree with at most $2^{|V_u|-k}$
  different complete assignments to the universal
  variables. \autoref{lemma:accordance} gives that for any of the
  $2^{|V_u|}$ total assignments to $V_u$ there must be a corresponding
  leaf-term. Averaging gives that $\pi$ has at least
  $\frac{2^{|V|}}{2^{|V|-k}}=2^k$ leafs.
\qed
\end{proof}

\autoref{theorem:k} gives us a powerful method of constructing
formulas with large term-resolution proofs. It is sufficient to
construct a true QCNF whose models have many universal literals. Let
us construct one such formula. For a given parameter
$N\in\mathbb{N}^+$ construct the following formula with $2N$ variables
and $2N$ clauses.


 \begin{equation}\label{equation:iff}
   \forall x_1,\exists y_1,
   \dots,
   \forall x_N,\exists y_N.\, \bigwedge_{i\in 1..N} (\bar x_i\lor y_i)\land (x_i\lor\bar y_i)
 \end{equation}

 \begin{proposition}
    Any term-resolution proof of~\eqref{equation:iff} is exponential in $N$.
 \end{proposition}
 \begin{proof}
   Formula~\eqref{equation:iff} is true as each of the existential
   variables $ y_i $ can be set to the same value as the variable
   $x_i$ and thus satisfying the matrix.

   Let $\psi $ denote the matrix of~\eqref{equation:iff}.  Each pair
   of clauses $\bar x_i\lor y_i$ and $\bar y_i\lor x_i$ must be
   satisfied by any model $\tau$ of $\psi$, which can be only done in
   two ways: the model will contain the literals $\{y_i,x_i\}$ or the
   model contained literals $\{\bar y_i,\bar x_i\}$. Hence $\tau$
   contains a literal for each $x_i$ and for each $y_i$.
   \autoref{theorem:k} gives that at least $2^N$ models are needed in
   the leafs of any term-resolution proof.
\qed
 \end{proof}
\section {A Possible Remedy---Negation}
This section suggests a possible remedy to the weakness exposed in the
previous section. Instead of proving a formula true by
term-resolution, we propose to refute its negation by Q-resolution (an
analogous approach to the one of propositional resolution).

To construct a negation of a formula, we follow the standard equalities
$\lnot\forall x.\ \Psi=\exists x.\ \lnot\Psi $ and $\lnot\exists x.\
\Psi=\forall x.\ \lnot\Psi $.  In order to bring the matrix back to
conjunctive normal form, we add additional ({\em Tseitin})
variables~\cite{tseitin68}.  We use the optimization by
Plaisted-Greenbaum, which enables encoding variables' semantics only
in one direction~\cite{DBLP:journals/jsc/PlaistedG86}. In particular,
for each clause we introduce a fresh variable that is forced to true
when that clause becomes true. Using these variables, we construct a
clause that requires that at least one of the clauses is false.

It would be correct to insert these fresh variables at the end of the
prefix (existentially quantified) but we will see that it is useful to
insert them further towards the outer levels, if possible.

\begin{definition}\label{definition:negation}
  The negation of a formula $P.\ \phi$ is denoted as $\Neg(P.\ \phi)$
  and constructed as follows.
  For each clause $ C $ introduce a fresh variable $ n_C $.  Construct
  the prefix of $\Neg(P.\ \phi)$ from $P$ inverting all the
  quantifiers in $P$ and inserting each of the variables the variable
  $n_C$ after the variable with maximal level in $ C $.  Construct a
  matrix of $\Neg(P.\ \phi)$ as the following clauses.

\[
 \comprehension{ \comprehension{\bar l\lor n_C}{l\in C} }{C\in\phi}
\;\union\;
\{ \bigvee_{C\in\phi}\bar n_C \} 
\]
\end{definition}

\begin{example}
 The $\Neg\left(\forall x\exists y\exists z.\ (\bar x\lor y)\land(x\lor z)\right)$
 is equal to 
 $\exists x\forall y\exists c_1\forall z\exists c_2.\
         (x\lor c_1)\land
         (\bar y\lor c_1)\land
         (\bar x\lor c_2)\land
         (\bar z\lor c_2)\land
         (\bar c_1\lor\bar c_2)$.
\end{example}

Clearly, $\Neg(\Psi)$ is false if and only if $\Psi $ is true. We say
that a QCNF $\Psi $ is {\em negation-refuted} by a
Q-resolution proof $\pi $ iff $\pi$ is a refutation of $\Neg(\Psi)$.

\subsection {Negation-Refutation P-simulates Term-Resolution}

The first question we should ask is whether for any term-resolution
proof there is a polynomial-size negation-refutation proof.  We show
this is indeed the case.

\begin{theorem}\label{theorem:simulation}
  Negation-refutation p-simulates term-resolution.
\end{theorem}
\begin{proof}[sketch]
Let $\pi$ be a term-resolution of a QCNF $P.\ \phi$.  Construct a
Q-resolution of $\Neg(P.\ \phi)$ as follows.
Let $M$ be a leaf of $\pi$.  From the rules of term-resolution for
each $C\in\phi$ there is a literal $l$ s.t.\ $l\in C$ and $l\in
M$. From the definition of $\Neg$, the QCNF $\Neg(P.\ \phi)$ contains
the clause $\bar n_C\lor\bar l$ for such literal.

Starting with the clause $\bigvee_{C\in\phi}\bar n_C$, resolve each
literal $\bar n_C$ with the clause $n_C\lor\bar l$, for each $l$ s.t.\ $l\in C$ and
$l\in M$. This results in the clause $\bigvee_{l\in M}\bar l$.  Note
that this clause does not contain contradictory literals because $M$
must not contain contradictory literals.

Repeating this process for each leaf of $\pi $ produces clauses that
are negations of those leaves. Perform Q-resolutions steps and
$\forall$-reductions as are done term-resolutions steps and
$\exists$-reductions in $\pi$.  This produces a proof where each node
to a negation of the corresponding node in $\pi$.  Since $\pi$ has the
empty term in the root, the produced tree has the empty clause in the
root.
Resolutions needed to produce each of the leaf clauses requires at most
$\min(|\pi|,|\var(\Phi)|)$ steps thus the resulting Q-resolution
is at most of size $(|\Phi|+|\pi|)^2$.
\qed
\end{proof}

\subsection{Separation Between Term-Resolution and Negation-Refutation}\label{section:separation}

The previous section shows that  negation-refutation is at least as
powerful as term-resolution. To show that the negation-refutation
proof system is in fact stronger, we recall
formula~\eqref{equation:iff}, whose term-resolution proofs are
exponential, and show it has a negation-refutation proof of linear
size.

 \begin{figure}[t]
   \centering
\begin{tikzpicture}[baseline=0,grow=up,xscale=3,yscale=.8]
\node(r) at (0,0)  {$D_{N-2}$  } ;
\node(n1) at (1,1) {$D_{N-2}\lor\bar x_N$} ;
\node(n2) at (1,2) {$D_{N-2}\lor\bar x_N\lor\bar y_N$ } ;
\draw[uredge] (n2)--(n1);
\node(n3) at (.75,3) { $D_{N-2}\lor\bar y_N\lor\bar c_{2N}$ } ;
\node(n4) at (1.66,3) {$\bar x_N\lor c_{2N}$};
\draw[redge] (n3)--(n2)--(n4);
\node(n5) at (1,4) {$\bar y_N\lor c_{2N-1}$};
\node(t) at (0,4) { $D_{N-2}\lor\bar c_{2N-1}\lor\bar c_{2N}$ } ;
\draw[redge] (t)--(n3)--(n5);
\node(n6) at (-1,1) {$D_{N-2}\lor x_N$} ;
\node(n7) at (-1,2) {$D_{N-2}\lor x_N\lor y_N$ } ;
\draw[uredge] (n7)--(n6);
\node(n8) at (-.75,3) { $D_{N-2}\lor y_N\lor\bar c_{2N-1}$ } ;
\node(n9) at (-1.66,3) {$x_N\lor c_{2N-1}$};
\draw[redge] (n8)--(n7)--(n9);
\node(n10) at (-1,4) {$y_N\lor c_{2N}$};
\draw[redge] (n10)--(n8)--(t);
\draw[redge] (n6)--(r)--(n1);
%
%
\end{tikzpicture}
   \caption{Resolving away negation of~\eqref{equation:iff} (where 
      $D_{N-2}=\bar c_1\lor\dots\lor\bar c_{2N-3}\lor \bar c_{2N-2}$).
    }
   \label{fig:linear}
 \end{figure}
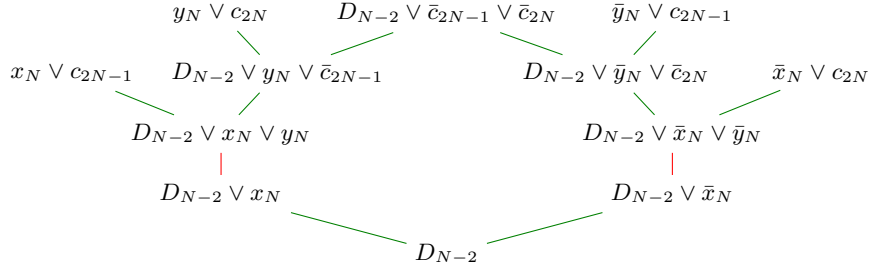

 \begin{proposition}
   Formula \eqref{equation:iff} has a linear 
   negation-refutation proof.
 \end{proposition}
 \begin{proof}[sketch]
   Negation of \eqref{equation:iff} introduces variables
   $c_1,\dots,c_{2N}$ representing the respective clauses.
   In particular, the following clauses are constructed
       $     x_i\lor c_{2i-1}$,
       $\bar y_i\lor c_{2i-1}$,
       $\bar x_i\lor c_{2i}$,
       $     y_i\lor c_{2i}$ for $i\in 1..N$ and the clause 
   $\bar c_1\lor\dots\lor\bar c_{2N}$. With the prefix 
   $\exists x_1\forall y_1\exists c_1c_2\dots
    \exists x_1\forall y_1\exists c_{2N-1}c_{2N}$.

    We show how to resolve away the literals $\bar c_{2N-1}$ and $\bar
    c_{2N}$; the rest of the $c_i$ literals is resolved in the same
    fashion.  For conciseness we define $D_{N-2}$ as $\bar
    c_1\lor\dots\lor\bar c_{2N-3}\lor \bar c_{2N-2}$.
    \autoref{fig:linear} shows how $\bar c_{2N-1}$ and $\bar c_{2N}$
    are replaced by $y_N$ and $x_N$ at which point$ y_N $ is universally reduced. Analogously,
    the literals are replaced with $\bar x_N$ and $\bar y_N$, which enables resolving
    $x_N$ away.

    Using this construction, each of the literals $\bar c_{2i-1},\bar c_{2i}$ are
    resolved away in $7$~resolution/reduction steps thus resulting in a
    resolution proof with $7N$ resolution/reduction steps in total.
\qed
 \end{proof}

\subsection{Variable Definitions}
\label{sec:definitions}

We observe that formula~\eqref{equation:iff} is an example of a formula
where an existential variable $y$ is {\em defined}, i.e.\ the value of
the variable is determined by values of some variables with a lower
level (in the case of formula ~\eqref{equation:iff} the value of $y_i$ is
determined by the value of $x_i$). So the natural question to ask is
whether any definition can be proven true by negation-refutation.
We show that this is indeed the case but we will need {\em
  QU-resolution}---an extension of Q-resolution that enables resolving
on {\em universal variables}~\cite{DBLP:conf/cp/Gelder12}.

We will demonstrate how negations of definitions can be refuted on the
following representative example.  Consider the prefix $\exists
x_1\forall x_2\exists x_3o_1o_2o_3$ and a matrix capturing the
equalities $o_1=\nand(x_1,x_2)$, $o_2=\nand(x_2,x_3)$, and
$o_3=\nand(o_1,o_2)$. These correspond to the following clauses
(Tseitin variables that will be used for negating the clauses are indicated in
parentheses).

\[
\begin{array}{lclcl}
  (c_1)\; \bar x_1\lor \bar x_2 \lor \bar o_1 &\;\;\;&
  (c_4)\; \bar x_2\lor \bar x_3 \lor \bar o_2 &\;\;\;&
  (c_7)\; \bar o_1\lor \bar o_2 \lor \bar o_3 \\
  (c_2)\; x_2\lor o_1&\;\;\;&
  (c_5)\; x_2\lor o_2&\;\;\;&
  (c_8)\; o_1\lor o_3\\
  (c_3)\; x_2\lor o_1&\;\;\;&
  (c_6)\; x_2\lor o_2&\;\;\;&
  (c_9)\; o_1\lor o_3\\
\end{array}
\]

After negating this formula, we obtain the following prefix.

\[\forall x_1\exists x_2\forall x_3
\forall o_1
\exists c_1c_2c_3
\forall o_2
\exists c_4c_5c_6
\forall o_3
\exists c_7c_8c_9\]

We omit the negated formula's matrix for succinctness.  The
Q-resolution proof proceeds in a similar fashion as the one
for~\eqref{equation:iff}.  Starting with the clause $\bar
c_1\lor\dots\lor\bar c_9$, the $\bar c_i$ literals are resolved away,
starting with the innermost ones.

\autoref{fig:definition} shows a fragment of the proof, which resolves
away the literals $\bar c_7,\dots,\bar c_9$ (certain resolution steps
are collapsed).  Using the clauses determining the value of $o_3$, the
proof generates the clauses {$\bar c_1\lor \dots\lor \bar o_1\lor\bar
  o_2$} and {$\bar c_1\lor\dots\lor o_2$}. Resolving these two clauses
removes the variable $o_2$.  Note that $ o_2 $ is universal, which is
why we need QU-resolution. In order to
resolve away $o_1$, the clause {$\bar c_1\lor\dots\lor o_1$} is
generated analogously.  Leaving us with the clause $\bar
c_1\lor\dots\lor\bar c_6$.  Note that it was possible to $\forall$-reduce
$o_3$ throughout the process because it is blocked only by the variables
$c_7,\dots,c_9$.  In contrast, the variables $o_1$ and $o_2$ could
{\em not} be $\forall$-reduced because they are blocked by 
the literals $\bar c_5,\dots,\bar c_6$.
The literals $\bar
c_4,\dots,\bar c_6$ and subsequently $\bar c_1,\dots,\bar c_3$ our
resolved in the same fashion.

\begin{figure}[t]
  \centering
\begin{tikzpicture}[xscale=1.2,yscale=.9]
\node (t) at (3.5,10)  { $\bar c_1\lor\dots\lor\bar c_6\lor\bar c_7\lor\bar c_8\lor\bar c_9$ }  ;
\node (a) at (-.7,10)  { $o_3\lor c_7$ };
\node (b) at (-.7,9.5)  { $\bar o_1\lor c_8$ };
\node (d) at (-.7,9)  { $\bar o_2\lor c_9$ };
\node[below] (e) at (2,9.5)  { $\bar c_1\lor\dots\lor\bar c_6\lor\bar o_1\lor\bar o_2\lor o_3$ };
\draw[redge] (a) to[out=-25,in=-185] (e.north);
\draw[redge] (b) -- (e.north);
\draw[redge] (d) to[out=25,in=185] (e.north);
\draw[redge] (t) -- (e.north);
\node[below] (er) at (2,8.7)  { $\bar c_1\lor\dots\lor\bar c_6\lor \bar o_1\lor\bar o_2$ };
\draw[redge][uredge] (e.south) -- (er) ;
\node (n1) at (7.3,10) { $\bar o_3\lor c_8$ };
\node (n2) at (7.3,9.5)  { $\bar o_3\lor c_8$ };
\node (n3) at (7.3,9)  { $o_2\lor c_7$ };
\node[below] (e1) at (5,9.5)  { $\bar c_1\lor\dots\lor\bar c_6\lor o_2\lor\bar o_3$ };
\draw[redge] (n1) to[out=205,in=5] (e1.north);
\draw[redge] (n2) -- (e1.north);
\draw[redge] (n3) to[out=-205,in=-5] (e1.north);
\draw[redge] (t) -- (e1.north);
\node[below] (e1r) at (5,8.7)  { $\bar c_1\lor\dots\lor\bar c_6\lor o_2$ };
\draw[uredge] (e1.south) -- (e1r) ;
\node (b) at (3.5,7.7)  { $\bar c_1\lor\dots\lor\bar c_6\lor\bar o_1$ }  ;
\draw [redge]  (e1r) -- (b.north) -- (er) ;
\end{tikzpicture}
  \caption{Resolving definitions}
  \label{fig:definition}
\end{figure}
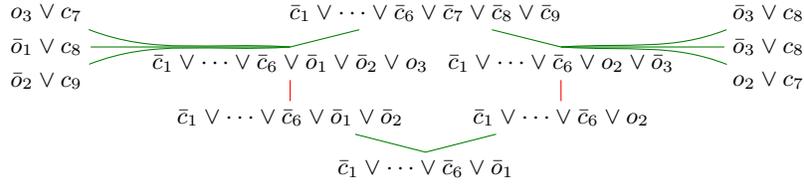

An analogous proof can be carried out for any acyclic circuit of
$\nand$ gates. One picks a topological order of the gates and resolves
them away as in the example above.

\section{ Summary, Conclusions, and Future Work}
\label{sec:conclusion}
 
This paper investigates the strength of term-resolution: a
well-established calculus for true quantified Boolean formulas.  This
paper exposes a significant vulnerability in the term-resolution
calculus, which stems from the fact that the number of leafs of a
term-resolution proof is not bound by the size of the formula in
question. Instead, the model-generation rule enables generating new
leafs of the proof from models of the matrix. The paper demonstrates
that this lets us force the proof to generate exponentially many leafs
by constructing QBF matrices with ``many'' universal literals.

This theoretical observation provides a further underpinning of the
well-known observation that solving quantified Boolean formula with a CNF
matrix can be sometimes particularly
harmful~\cite{DBLP:conf/aaai/AnsoteguiGS05,DBLP:conf/aaai/Zhang06}.
Indeed, we demonstrate that even a very simple formula where each
clause has only two literals leads to exponential term-resolution
proofs.

At the practical level, in response to this issue, Zhang proposes to reason
on a formula and on its negation at the same
time~\cite{DBLP:conf/aaai/Zhang06}. This idea was realized with
different flavors in various
solvers~\cite{DBLP:conf/aaai/GoultiaevaB10,DBLP:conf/sat/KlieberSGC10,GoultiaevaEtAl13}.
The second part of this paper takes a similar avenue at the
theoretical level.  We compare the term-resolution calculus with the
negation-refutation calculus, a calculus which refutes the formula's
negation in order to show the formula true.  The paper demonstrates
that this proof system indeed has favorable theoretical properties, in
particular it p-simulates term-resolution and there is an exponential
separation between the two calculi.

This result is related to the well-known fact that enabling adding new
variables in propositional resolution yields a more powerful proof
system (extended resolution)~\cite{Cook:1976:SPP:1008335.1008338}.
Negation-refutation  introduces new variables  too.  However, in
extended resolution, the prover must come up with the variables'
definitions. In negation-refutation, the definitions are determined by
the clauses of the formula.

The last part of the paper touches upon some limitations of the
negation-refutation calculus. If a variable's value is defined as a
function of some other variables, through a Boolean circuit, we ask if
it's possible to prove that it is always possible to come up with the
right value for the variable being defined, i.e.\ complete the
circuit. This is something that we would hope to be proven easily. We
show that it is indeed possible to prove such definitions true lineary
using negation-refutation but we show so with the use of
QU-resolution---extension of Q-resolution that enables resolving on
universal variables. This result is important from a theoretical
perspective but raises further questions because existing QBF solvers
use Q-resolution. It is the subject
of future work to look for linear proofs for such formulas using only
Q-resolution.



\end{document}